\documentclass[sigconf]{acmart}
\usepackage{dsfont}
\usepackage{booktabs} 
\usepackage{multirow}
\usepackage{amsmath}
\usepackage{colortbl}
\usepackage{indentfirst}
\usepackage{amssymb}
\usepackage{array}
\usepackage{graphicx}
\usepackage{epstopdf}
\usepackage{caption}
\usepackage{subcaption}
\usepackage{algorithmic}
\usepackage{algorithm}
\usepackage{amsmath}
\usepackage{graphicx}
\usepackage{bbm}
\usepackage{bm}
\newtheorem{assumption}{\bf Assumption}
\AtBeginDocument{%
  \providecommand\BibTeX{{%
    \normalfont B\kern-0.5em{\scshape i\kern-0.25em b}\kern-0.8em\TeX}}}

\copyrightyear{2019} 
\acmYear{2019} 
\setcopyright{acmcopyright}
\acmConference[Mobihoc '19]{The Twentieth ACM International Symposium on Mobile Ad Hoc Networking and Computing}{July 2--5, 2019}{Catania, Italy}
\acmBooktitle{The Twentieth ACM International Symposium on Mobile Ad Hoc Networking and Computing (Mobihoc '19), July 2--5, 2019, Catania, Italy}
\acmPrice{15.00}
\acmDOI{10.1145/3323679.3326524}
\acmISBN{978-1-4503-6764-6/19/07}

\begin{document}

\title{Quantized VCG Mechanisms for Polymatroid Environments }
\titlenote{This work was supported in part by the National Science Foundation grant TWC-1314620.}

\author{Hao Ge and Randall A. Berry}
\affiliation{%
	\institution{Department of Electrical and Computer Engineering, Northwestern University, Evanston, Illinois}
}


\begin{abstract}
Many network resource allocation problems can be viewed as allocating a divisible resource, where the allocations are constrained to lie in a polymatroid. We consider market-based mechanisms for such problems. Though the Vickrey-Clarke-Groves (VCG) mechanism can provide the efficient allocation with strong incentive properties (namely dominant strategy incentive compatibility), its well-known high communication requirements can prevent it from being used. There have been a number of approaches for reducing the communication costs of VCG by weakening its incentive properties. Here, instead we take a different approach of reducing communication costs via quantization while maintaining VCG's dominant strategy incentive properties. The cost for this approach is a loss in efficiency which we characterize. We first consider quantizing the resource allocations so that agents need only submit a finite number of bids instead of full utility function. We subsequently consider quantizing the agent's bids. 
\end{abstract}



\maketitle

\section{Introduction}
Efficient allocation of limited resources is a fundamental problem in networked systems.  Economic-based models are a common way of addressing this problem.  In such models, 
agents are endowed with utility functions and the goal is to maximize the social welfare, given by the sum of the agents' utilities. Hence, given the agents' utility functions, efficient resource allocation then reduces to solving an optimization problem. In networked systems, two challenges to such an approach are (1) the agents are distributed and only limited  communication may be available to communicate user utilities and resource allocations decisions, and (2) self-interested users may not have incentive to correctly report their utilities. 

The incentive issue has been studied through the lens of mechanism design. One of the most celebrated results in this literature is the Vickrey-Clarke-Groves (VCG) mechanism \cite{vickrey1961counterspeculation,clarke1971multipart,groves1973incentives}, which provides an elegant solution to the incentive issue. In the VCG mechanism, each agent is required to report her utility function and in turn receives an allocation and makes a payment based on the reports of every agent. Through a carefully designed payment, the VCG mechanism achieves the efficient outcome under a strong incentive guarantee. Namely, it is dominant strategy incentive compatible (DSIC) for agents to truthfully report their utility and if they do so, then the mechanism simply maximizes the sum of the reported utility, yielding the efficient outcome. However, in the context of divisible network resources, such as power, bandwidth or storage space, each agent's utility may be infinite dimensional and hence VCG mechanisms do not address the issue of limited communication.

The limited communication issue has also been well studied. In particular, we highlight the seminal work of Kelly, \cite{kelly1997charging,kelly1998rate}, in which agents submit one-dimensional bids and receive an allocation determined by the ratio of their bid and a congestion price.  For price-taking agents, who do not anticipate the impact of their actions on the price, this mechanism can lead to the socially optimal allocation, but as shown in \cite{johari2004efficiency}, when buyers can anticipate the impact of their actions on congestion price, this mechanism can lead to an efficiency loss.

In \cite{yang2007vcg,johari2009efficiency}, a ``VCG-Kelly" mechanism was proposed that combined the one-dimensional bids in Kelly's mechanism and the VCG payment rule. This mechanism is shown to obtain the socially optimal outcome as with VCG. However, the incentive guarantee is relaxed to the weaker notion of Nash equilibria instead of a dominant strategy equilibria as in VCG.\footnote{Justifying a Nash equilibria requires that agents are aware of each others' pay-offs (and rationality) or apply some type of dynamic process to reach this point.}

In our previous work, \cite{haoge2018,ge2018quantized}, we introduced an alternative approach for reducing communication costs while maintaining VCG's dominant strategy incentive properties when allocating a single divisible resource. This approach was to quantize the resource and then run VCG on the quantized resource. This incurs some efficiency loss, which was bounded. We build on this approach in this paper, but instead of a single divisible resource, we consider resource allocations that are constrained to lie within a polymatroid in $\mathbb R_+^N$, where $N$ is the number of agents. Such a polymatroid environment is a natural and non-trivial generalization of a single divisible resource. For example, if the single divisible resource represents the rate on a given link in a network, then a polymatroid can be viewed as giving the region of rates obtainable by different source-destination pairs across a network with multiple links (see e.g. \cite{federgruen1986optimal}). As shown in \cite{Yang2005RevenueAS,berry2006kelly}, mechanisms for single links may not work when applied in a more general network setting. Polymatroid environments can also model many other network settings such as multiclass queueing systems \cite{shanthikumar1992multiclass}, various network flow problems \cite{lawler1982computing,imai1983network,lawler1982flow}, spectrum sharing \cite{kang2014sum,berry2013market} and problems in network information theory \cite{tse1998multiaccess,ge2015secure,tse2004diversity}. Given a polymatroid feasible region, we study quantized VCG mechanisms and the resulting efficiency losses for strategic agents. 

The contributions of this paper are the following: 
\begin{itemize}
	\item{Given a polymatroid environment, we determine a way to quantize the resource into partitions so that a feasible allocation of partitions lies in an {\it integral polymatroid}.}
	\item{Based on the integral polymatroid, we propose a \textit{quantized VCG mechanism} for allocating partitions.  This mechanism inherits VCG's DSIC property and further, due to the polymatroid structure, can be implemented using a greedy algorithm. We also characterize the worst-case efficiency loss of this mechanism due to the quantization.}
	\item{In addition to the quantization of the resource, we further study a  {\it rounded VCG mechanism} in which the agents' bids are also quantized. The bid quantization results in the mechanism no longer being DSIC, but we show that $\epsilon$-dominant strategies exist. We characterize the worst-case efficiency for three different types of such strategies.}
	\item{We analyze the trade-off between communication cost and overall efficiency in our proposed mechanisms. We further discuss the impact of the parameters in each mechanism.}
\end{itemize}

The rest of the paper is organized as follows: in Sect. 2,  we review some properties of polymatroids and the VCG mechanism. Our two classes of quantized mechanisms are introduced and analyzed in Sect. 3-4. We conclude the paper in Sect. 5.

\emph{Notations:} $\mathbb{R}_{+}^N$ and $\mathbb{N}^N$ denote the set of nonnegative real-valued and integer-valued vectors of dimension $N$, respectively. We use $\lfloor x \rfloor$ for the floor, $\lceil x \rceil$ for the ceiling, and $\{x\}$ for the fractional part of $x$. We denote by $\sigma_{-i}$ the subset of $\sigma$ for everybody except for agent $i$ and $(z,\sigma_{-i})$ denotes the set $\sigma$ with the $i^{th}$ component replaced with $z$. We write $[ \cdot]^+$ as the projection onto the nonnegative orthant. The indicator function $\mathds{1}_w$ equals 1 if the event $w$ happens, otherwise 0.

\section{Model and Background}
\subsection{Polymatroids}
We begin by reviewing some basic definitions and properties related to polymatroids. 

\begin{definition}\cite{edmonds1970submodular}
	Given a finite ground set $\mathcal{N} = \{1,\cdots,N\}$, a set function $f$: $2^{\mathcal{N}} \rightarrow \mathbb{R}$ is \textit{submodular} if:
	\begin{equation*}
		f(S) + f(T) \geq f(S \cap T) + f(S \cup T),
	\end{equation*}
	for any subsets $S$ and $T$ of $\mathcal N$. If the inequality holds strictly, the set function $f$ is strictly submodular.
\end{definition}

An equivalent definition for submodularity is:
\begin{equation*}
	f(S\cup \{i\})+f(S\cup\{j\}) \geq f(S) + f(S\cup\{i, j\}  ),
\end{equation*} 
for any $S \subseteq \mathcal{N}$ and distinct $i,j \in \mathcal{N} \setminus S$.

A set function $f$: $2^{\mathcal{N}} \rightarrow \mathbb{R}$ is normalized if $f(\emptyset) = 0$; monotone (non-decreasing) if $f(S) \leq f(T), \forall S \subseteq T$; and integer-valued if $f(S) \in \mathbb{Z}, \forall S \subseteq \mathcal{N}$.

\begin{definition}
	Let 
	\begin{equation*}
		P_f = \{{\bf x} \in \mathbb{R}_{+}^N|\sum_{n \in S}x_n \leq f(S), \forall S \in \mathcal{N}\}.
	\end{equation*}
	If the set function $f$ is normalized, monotone and submodular, then $P_f$ is a {\it polymatroid} with rank function $f$.
\end{definition} 

If the rank function of a polymatroid is integer-valued, then it follows that all of corner points are also integer valued. Further, it can sometimes be convenient to regard an \textit{integral polymatroid}, i.e., a polymatroid consisting of only points in $\mathbb N^N$ instead of points in $\mathbb R^N$ \cite{edmonds1970submodular}. We adopt this view here when discussing integral polymatroids. 

The following are two well known properties of polymatroids. 
\begin{lemma}\label{lemma:1}
	Given a polymatroid $P_f$, if ${\bf x^*} \in P_f$, then the polyhedron
	$$
	P_f -{\bf x}^* := \{{\bf x} \in \mathbb{R}_{+}^N | \sum_{n\in S} x_n \leq f(S)-\sum_{n \in S} x^*_n,  \forall S \in \mathcal{N}\}.
	$$
	is also a polymatroid with rank function $\tilde{f}$, where 
	$$
	\tilde{f}(S) = \min_{T:S\subseteq T}(f(T)-\sum_{n \in T} x^*_n).
	$$
\end{lemma}
\begin{lemma}\label{lemma:2.1}
	For any polymatroid $P_f$, there exists a ${\bf x} \in P_f$ such that $\sum_{n=1}^N x_n =f(\mathcal{N})$.
\end{lemma}

The second lemma shows that the sum constraint over all $\mathcal N$ is always tight for some point within the polymatroid. The set of all such points is referred to as the polymatroid's dominant face, which is formally defined next. 
\begin{definition}
	The \textit{dominant face} of a polymatroid $P_f$, denoted by $\mathcal{D}(P_f)$, is defined as:
	$$
	\mathcal{D}(P_f) = \{{\bf x} \in \mathbb{R}_{+}^N| {\bf x} \in P_f, \sum_{n=1}^N x_n = f(\mathcal{N}) \}.
	$$ 
\end{definition}

As defined in \cite{salimi2015representability}, the minimum distance of a set function $f$ over a ground set $\mathcal{N}$ is:
\begin{equation}
	\bigtriangleup f := \min_{S\subseteq \mathcal{N}} \min_{i,j \in \mathcal{N} \setminus S} |D_f(S \cup \{i\}, S\cup\{j\})|,
\end{equation}
where the distance $D_f(S,T)$ between two arbitrary subsets $S$ and $T$ is defined as:
\begin{equation}
	D_f(S,T) = f(S) + f(T) -f(S \cap T) -f(S\cup T).
\end{equation} 
This can be thought of as an indication of how strictly supermodular the set function is. For a submodular function $f$, the distance between two arbitrary subsets is always non-negative. Furthermore, if the set function $f$ is strictly submodular, then the minimum distance $\bigtriangleup f$ is strictly greater than 0.

Following \cite{salimi2015representability}, we have the following theorem:
\begin{theorem}
	Given a polymatroid $P_f$ with $\bigtriangleup f = 0$, then for any $\eta > 0$, there exists another polymatroid $P_{\tilde{f}}$ such that $\bigtriangleup \tilde{f} > 0$ and $f(S) -\tilde{f}(S) \leq \eta, \forall S \subseteq \mathcal{N}$.
\end{theorem}
\begin{proof}
	Let $g$ be a normalized set function defined as $g(S) = |S^c|^2$ for any nonempty set $S$. Then let $\tilde{f}(\emptyset) = 0$ and $\tilde{f}(S) = f(S) - \gamma g(S)$ for any nonempty set $S$. With a small value of $\gamma$, we can make sure the difference between $f$ and $\tilde{f}$ is less than $\eta$. Moreover, since $-g(S)$ is submodular with a nonnegative minimum distance, we can show $\tilde{f}$ is a rank function with $\bigtriangleup \tilde{f} > 0$. The detailed proof is omitted due to space consideration. 
\end{proof}

The above theorem shows that given any polymatroid, we can find another polymatroid with positive minimum distance such that the new polymatroid lies in the given one and the discrepancy could be arbitrarily small. Due to this reason, in the remainder of the paper, we only consider polymatroids constrained by set functions with positive minimum distance.\footnote{If this is not the case, then we can use the above construction and as will be evident it would have negligible impact on the resulting efficiency of our mechanisms.}

\subsection{Social Welfare Maximization Problem}
We consider a model where a set of $N$ agents are competing for a resource constrained by a polymatroid environment. Denote the set of agents as $\mathcal{N}= \{1,\cdots, N\}$. A feasible allocation ${\bf x} =\{x_n\}_{n=1}^N$ is then any point that lies in a polymatroid associated with a set function $f$ over $\mathcal{N}$. In particular, the agents in any subset $S$ could obtain a total allocation of at most $f(S)$, i.e., $\sum_{n \in S}x_n \leq f(S)$. Without loss of generality, we assume that $f(\mathcal{N}) = 1$. Moreover, each individual's allocation must be non-negative and hence  ${\bf x} \in \mathbb{R}_{+}^N$.

Each agent is endowed with a utility function $u_n(x_n)$, which only depends on the agent's final allocation. Throughout this paper, these utility functions satisfy the following assumptions:

\begin{assumption}
Each agent's utility function is only known by that agent, i.e., it is private information.
\end{assumption}

\begin{assumption}
The utility functions $u_n$ are nonnegative, concave, and strictly increasing. 
\end{assumption}

\begin{assumption}
For all agents, the marginal utility per unit resource is bounded, that is, there exist positive constants $\alpha$ and $\beta$ such that for any agent $n$, $\beta < u_n'<  \alpha$. These bounds are public information.  
\end{assumption}

\begin{assumption}
The utility functions are in monetary units and each agent's quasilinear payoff is defined as $u_n(x_n)-p_n$, where $p_n$ is the payment for the allocated resource.
\end{assumption}

Assumption 1 results in the incentive issue: utilities are private information, each agent may lie about them. Assumption 2 and 3 show agents have ``elastic" resource requirements and the unit value for the resource lies within some range, which is reasonable in many practical settings. The last assumption is widely adopted in pricing mechanisms and the quasilinear structure makes each agent's strategy and the corresponding final allocation dependent on the payment rule in this game as well as the utility functions.

Under this setting, the goal of the allocator is to make the best use of the resource, i.e., to maximize the social welfare. This is given by solving the following optimization problem:
\begin{align}\label{main:1}
&\max \ \ \ \ \sum_{n=1}^N u_{n}(x_n)\\
\nonumber &\text{subject to} \ \ {\bf x}\in P_f= \{{\bf x} \in \mathbb{R}_{+}^N | \sum_{n\in S} x_n \leq f(S),  \forall S \in \mathcal{N}\}.
\end{align} 

Since the total sum constraint is always binding in polymatroids, we have the following result from \cite{groenevelt1991two}. We give a proof of this for completeness.
\begin{proposition}\label{prop:1}
The optimal solution to optimization problem (\ref{main:1}) lies on the dominant face of the polymatroid.
\end{proposition}
\begin{proof}
We show this by means of contradiction. Suppose $\bf x^*$ is the optimal solution, but $\sum_{n=1}^N x^*_n < f(\mathcal{N})$. Then we can construct a new polymatroid $P_g$ by subtracting the optimal solution $\bf x^*$:
\begin{equation} 
	P_g = \{{\bf x} \in \mathbb{R}_{+}^N | \sum_{n\in S} x_n \leq f(S)-\sum_{n \in S} x^*_n,  \forall S \in \mathcal{N}\}.
\end{equation}
The new set function is $g(S) = f(S)-\sum_{n \in S} x^*_n$, which satisfies the properties for set function to define a polymatroid. Since $f(\mathcal{N}) -\sum_{n=1}^N x^*_n > 0 $, the polymatroid $P_g$ is not empty and so we can  find a nonzero vector $\bf x'$ inside. Moreover, each utility function is strictly increasing, hence $\sum_{n=1}^N u_n(x^*_n+x_n') >\sum_{n=1}^N u_n(x^*_n)$, which is a contradiction. This finishes the proof.
\end{proof}

\subsection{Vickrey-Clarke-Groves (VCG) Mechanism}
Next we review the VCG mechanism. In this mechanism, the allocator asks agents to report their utility functions and determines the allocation which maximizes the social welfare under the reported utilities. The VCG payments for given agent is then given by the difference between the total utility the other agents would have received if the given agent was not present and the total utility that they received based on the reported utilities.\footnote{This is the most common form of VCG payments known as the Clarke pivot rule, more generally changing agent $n$'s payment by any function that does not depend on the reported utility of agent $n$ is also valid.} In other words, each agent's payment equals the increase in the others' sum utility if she is absent. 

\begin{definition}\cite{nisan2007algorithmic}
A mechanism where players have private information is said to be \textit{dominant-strategy incentive-compatible} (DSIC) if it is a weakly-dominant strategy for every player to reveal his/her private information.
\end{definition}

The VCG mechanism is DSIC, and given that agents follow the dominant strategy of reporting their true utility, the mechanisms will be efficient, i.e. it will maximize the social welfare.

\subsection{Quantized VCG Mechanism }
Given a divisible resource, the utility function an agent submits in the VCG mechanism can be an arbitrary real valued function of the amount of resource obtained. As we have discussed, this can result in the excessive communication overhead. Next we discuss the approach in \cite{haoge2018} for limiting this overhead via quantization, when the resource is a single divisible resource, i.e., the resource constraint is given by $\sum_{n\in\mathcal N} x_n \leq 1$.  In this approach, the allocator partitions the whole resource into $M$ divisions equally and restricts agents to get an integer number of divisions. Essentially, this mechanism then runs a VCG mechanism for the quantized resource. Under this setting, each agent $n$ only needs to submit $M$ bids: $v_{n1}, \cdots, v_{nM}$, where $v_{nm}$ indicates agent $n$'s marginal utility from getting 
an $m$th additional unit of resource. This reduces the information an agent needs to submit from specifying an 
infinite dimensional function to specifying an $M$-dimensional vector.  For small values of $M$, this quantization also simplifies the optimization faced by the resource allocator as it can now determine the socially optimal allocation via a greedy algorithm instead of needing to solve a convex optimization problem.\footnote{Also note that using a greedy algorithm gives the exact social optimal for the quantized problem, while when solving the convex optimization, the solution will typically be with some $\epsilon$ specified by the algorithm used. For large $\epsilon$, this in turn could impact the incentive guarantee of VCG.}  Finally, since this is still a VCG mechanism, it it still DSIC and implements the social optimal outcome with the quantized resource. However there is a loss in efficiency due to the quantization of the resource.  In \cite{haoge2018}, tight bounds on this loss are given.  Specifically, the efficiency (defined as the ratio of the welfare with $M$ divisions to the optimal welfare without quantization) is shown to always be no less than   
\begin{equation}\label{eq:old}
\frac{M}{M+N-1}
\end{equation}
where $N$ is the number of agents. Here we extend this to polymatroid environments.

\section{Quantizing VCG for Polymatroids }
In this section, we give two toy examples that show blindly partitioning the resource into  devisions may result in a large efficiency loss or the failure of greedy algorithm for determining the assignment.

\begin{example}
	Consider two agents competing for the resource, where the feasible region is a polymatroid $P_f = \{(x_1,x_2) \in \mathbb{R}^2_+| x_1 \leq 0.6 , x_2 \leq 0.6, x_1 + x_2 \leq 1\}$, their utilities are $u_1(x_1) = 2x_1$ and $u_2(x_2) = x_2$, respectively.
\end{example}

For the setting in Example 1, suppose we partition the total resource into 3 divisions equally ($M = 3$), so that the amount of each division is $\frac{1}{3}$. The feasible region for the number of divisions each agent can get then becomes:\footnote{Here we are assuming that any quantized allocation has to respect the original constrains in $P_f$; hence, for example, agent 1 can not receive two divisions since $2/3 > 0.6$.}
$$
P_{\tilde{f}} = \{(y_1,y_2) \in \mathbb{N}^2| y_1 \leq 1 , y_2 \leq 1, y_1 + y_2 \leq 3\}.
$$
Observe that here the constraint function is no longer submodular and the sum constraint in  $P_{\tilde{f}}$ is never tight. In particular note that the maximum number of divisions allocated is now 2 ($y_1 +y_2 \leq 2$), which means that $1/3$ of the total resource is never allocated. For the given utilities this results in an efficiency of $0.63$, which is lower than the value of $0.75$ given by the bound in (\ref{eq:old}).\footnote{Moreover this poor efficiency  is obtained with linear utilities, which lead to no loss in the setting in \cite{haoge2018}.}  By changing the parameters in this example, even lower efficiencies are possible. 

\begin{example}
	Consider the following polymatroid as the feasible region: $P_f =\{ {\bf x} \in \mathbb{R}^3_+| x_1 \leq 0.7, x_2\leq 0.7, x_3\leq 0.7, x_1 +x_2 \leq 0.9, x_2+x_3 \leq 0.9, x_1+x_3 \leq 0.9, x_1+x_2+x_3 \leq 1  \}$, the utility functions for the agents are $u_1(x_1) = 1.2 x_1$, $u_2(x_2) = 1.1x_2$, $u_3(x_3) = x_3$.
\end{example}

Similarly, if we partition the resource in Example 2 into 3 divisions, then the feasible region for number of divisions each agent can get is:
\begin{align*}
	P_{\tilde{f}} =  & \{(y_1,y_2, y_3) \in \mathbb{N}^3| y_1 \leq 2 , y_2 \leq 2,y_3 \leq 2, y_1 + y_2 \leq 2,\\
	& y_1+y_3 \leq 2, y_2+y_3 \leq 2, y_1+y_2+y_3 \leq 3\}.
\end{align*}
In this case ${\bf y} = (1,1,1)$ is a feasible solution and the corresponding allocation is ${\bf x} = (\frac{1}{3}, \frac{1}{3}, \frac{1}{3})$. In this case, all of the resource is allocated and the social welfare is $1.1$. However, if we use the greedy algorithm to find the optimal allocation given the true marginal valuations, the solution and the corresponding allocation are ${\bf y}^* = (2,0,0)$ and ${\bf x}^* = (\frac{2}{3},0,0)$, respectively. This results in a social welfare of $0.8$ which is less than the value of $1.1$ obtained by the feasible solution $(1,1,1)$. In other words the greedy algorithm no longer gives the optimal quantized allocation. The issue again is that $P_{\tilde{f}}$ is not a polymatroid.

These two examples illustrate that in a polymatroid environment more care is needed when quantizing the resource.  In both cases, the issue is that region resulting from quantization is no longer a polymatroid. Next we discuss an approach  to ensure that this does not occur by ensuring that the region $P_{\tilde{f}}$ obtained after quantization is always an integral polymatroid.  

\subsection{Integral Polymatroid Construction}
Given a real-valued polymatroid, $P_f$,we next show that if a large enough  number of partitions are used, then the issues in the previous two examples do not arise. The construction for doing this is given in Algorithm \ref{integer-valued}. This algorithm specifies an integral set $P_{\tilde{f}}$, which indicates the number of quantized partitions for each agent.

\begin{algorithm}[h]
	\caption{Resource Partition}
	\begin{algorithmic}\label{integer-valued}
		\STATE \textbf{Inputs:} Polymatroid $P_f$ associated with set function $f$, which has the minimum distance $\triangle f$.
		\STATE	Choose integer $M$ such that $ M \geq \frac{2}{\triangle f}$.
		\FOR {$S \subseteq \mathcal{N}$ }
		\STATE {$\tilde{f}(S) = \lfloor f(S)M \rfloor$;}
		\ENDFOR 
	\end{algorithmic}
\end{algorithm}

We next give several properties of the set $P_{\tilde{f}}$ given by Algorithm 1.  First, it is straightforward to see the following lemma, which states that any constraint in $P_{\tilde{f}}$ is within $\frac{1}{M}$ of the corresponding constraint in $P_f$ and that the sum constraints are the same.
\begin{lemma}
	For any $S \subseteq \mathcal{N}$, $0 \leq f(S) - \frac{\tilde{f}(S)}{M} < \frac{1}{M}$; further, $\frac{\tilde{f}(\mathcal N)}{M} = f(\mathcal N)=1$.
\end{lemma}

Next we show that the constraints in $P_{\tilde{f}}$ always allow any set of users to have at least two partitions. (Note this is not the case for Example 1 in the previous section when $M=3$). 

\begin{proposition}\label{prop:2}
	For any non-empty subset $S \subseteq \mathcal{N}$, $\tilde{f}(S) \geq 2$.
\end{proposition}
\begin{proof}
	First, according to the definition of the minimum distance, for any agent $i$ and $j$, we have:
	$$
	\bigtriangleup f \leq  f(\{i\}) + f(\{j\}) -f(\emptyset) -f(\{i,j\})\leq f(i).
	$$
	Hence, for any nonempty subset $S$, we can find $i\in S$ and
	\begin{equation*}
		\tilde{f}(S) = \lfloor f(S)\frac{2}{\bigtriangleup f}\rfloor \geq\lfloor f(\{i\})\frac{2}{\bigtriangleup f}\rfloor \geq 2.
	\end{equation*}
\end{proof}

Finally we show that $P_{\tilde{f}}$ is an integral polymatroid.

\begin{theorem}
	The polyhedron, $P_{\tilde{f}}$, associated with $\tilde{f}$ given by Algorithm 1 is an integral polymatroid.
\end{theorem}
\begin{proof}
	First, we know $\tilde{f}(\emptyset) = \lfloor f(\emptyset )M\rfloor = 0$, so $\tilde{f}$ is normalized.
	
	For monotonicity, consider two sets $S$ and $T$ such that $S \subseteq T$, since $f$ is monotonically increasing, then
	\begin{equation}
		\tilde{f}(S) = \lfloor f(S )M\rfloor \leq \lfloor f(T )M\rfloor = \tilde{f}(T).
	\end{equation}
	
	Moreover, we have 
	\begin{align}
		\nonumber &\tilde{f}(S+ \{e_1\})+\tilde{f}(S+\{e_2\}) \\
		\nonumber\geq &Mf(S+ \{e_1\})-1 + Mf(S+\{e_2\}) -1\\
		\nonumber\geq & M(\bigtriangleup f +f(S) +f(S+\{e_1\} +\{e_2\} )) -2\\
		\nonumber\geq & M\bigtriangleup f -2 +\tilde{f}(S) + \tilde{f}(S+\{e_1\} +\{e_2\} )\\
		\geq & \tilde{f}(S) + \tilde{f}(S+\{e_1\} +\{e_2\} ),
	\end{align}
	where we used the definition of minimum distance and the fact that $M\geq  \frac{2}{\bigtriangleup f} $, respectively, in the last two steps. Therefore, $\tilde{f}$ is submodular and $P_{\tilde{f}}$ is an integral polymatroid.
\end{proof}

Since $P_{\bar{f}}$ is a polymatroid, there always exists an allocation on its dominant face and  by construction $\tilde{f}(\mathcal N) = M$, so that any such allocation will utilize all of the available resource.

\subsection{Quantized VCG Mechanism}
Given the feasible region $P_f$, we define a {\it quantized VCG mechanism} as a mechanisms in which  the agents are only allocated an integer number divisions, where the set of allocations must lie in the integral polymatroid $P_{\tilde{f}}$ constructed in Algorithm 1. 

Under this setting, the optimal allocation is ${\bf x}^* = \frac{1}{M}  {\bf y}^* $, where ${\bf y}^*$ is the optional allocation of divisions given by the following integer optimization problem:
\begin{align}\label{5}
	&\text{maximize} \ \ \ \ \sum_{n=1}^N u_{n}(\frac{y_n}{M})\\
	\nonumber &\text{subject to} \ \ \ \ {\bf y} \in P_{\tilde{f}}.
\end{align}
\begin{corollary}\label{lemma:2}
	The maximizer ${\bf y^*}$ lies on the dominant face of $P_{\tilde{f}}$.
\end{corollary}

This is a straightforward generalization of Proposition 1 to the discrete setting and can be proven in a similar manner. Note also from the discussion above, such a solution will utilize all of the resource. 

To determine the allocation, the resource allocator essentially runs a VCG mechanism for the possible allocations in $P_{\tilde{f}}$. This requires each agent to submit its valuation for each possible bundle of units it can obtain. Since each agent $n$ can get at most $\tilde{f}(\{n\})$ units, agent $n$ only needs to report $\tilde{f}(\{n\})$ values: $u_n(\frac{y_n}{M})$, $1\leq y_n \leq \tilde{f}(\{n\})$. Equivalently, agent $n$ can submit the marginal utility instead. Note that the range of the marginal values will be smaller than the range of the actual utility values and so in some sense reporting the marginal values also reduces the communication cost. This will be made more precise in the next section, when we also consider quantizing the bids.  The marginal utility is denoted by $V_n = \{v_{n1},\cdots, v_{n\tilde{f}(\{n\})} \}$, where 
\begin{equation}
	v_{nm}  =  u_n(\frac{m}{M}) - u_n(\frac{m-1}{M}), \ \ \ \ 1 \leq m \leq \tilde{f}(\{n\}.
\end{equation}
Furthermore, under Assumption 3, all utility functions have bounded marginal valuations and hence each bid $v_{nm}$ is lower bounded by $\frac{\beta}{M}$. To further reduce the range of bids, we can require agent $n$ to submit the surrogate marginal utility $\hat{v}_{nm} = v_{nm} - \frac{\beta}{M}$ so that:
\begin{equation}
	u_n(\frac{y_n}{M}) = \sum_{m=1}^{y_n} \hat{v}_{nm} + \frac{\beta y_n}{M}.
\end{equation}
Notice that both $v_{nm} $ and $\hat{v}_{nm} $ are non-increasing in $m$ due to the concavity of $u_n$.

The detailed mechanism is as follows:
\begin{itemize}
	\item{Determine the number of partitions and partition the resource into $M$ divisions equally using Algorithm \ref{integer-valued}. The feasible region for number of units is $P_{\tilde{f}}$}.
	\item{Solicit and accept sealed surrogate marginal utility vectors $\hat{{V}}_n = \{\hat{v}_{n1}, \cdots, v_{n\tilde{f}(\{n\})} \}$}.
	\item{Determine the allocation of divisions that optimize (\ref{5}) where the objective is the utility given by the given bids $\hat{V}_n$. The final allocation agent $n$ gets is $y^*_n$. The sum of bids picked is denoted by 
		\begin{equation}
			\text{OPT}( \hat{V}) = \sum_{n=1}^N \sum_{m=1}^{y_n^*} \hat{v}_{nm}.
		\end{equation}
	}
	\item{Set price $ p_n$ for agent $n$ as:
		\begin{align*}
			p_n=\text{OPT}(0,\hat{V}_{-n})-(\text{OPT}({\hat{V}})-\sum_{m=1}^{y^*_n}\hat{v}_{nm}) + \frac{y^*_n\beta}{M}.
	\end{align*}}
\end{itemize}
We next give a greedy algorithm in Algorithm \ref{greedy} to determine the allocation in (\ref{5}).
\begin{algorithm}[h]
	\caption{Greedy allocation algorithm}
	\begin{algorithmic}\label{greedy}
		\STATE \textbf{Inputs:} Integer-valued polymatroid $P_{\tilde{f}}$; $N$ descending list: $\hat{{V}}_1, \cdots,\hat{{V}}_N$, the first element of $\hat{{ V}}_n$ is denoted by $\hat{{ V}}_n[0].$
		\STATE \textbf{Initialization:} Set ${\bf y} = {\bf 0} $, $\mathcal{N} = \{1,\cdots,N\}, $$J({\bf y}) = \mathcal{N}$.
		\WHILE {$J({\bf y})$ is not empty}
		\STATE {$n = \arg\max_{i \in \mathcal{N}}\{\hat{{ V}}_i[0]\}$ (break ties arbitrarily);}
		\STATE {$y_n$ $+=$ $1$;}
		\STATE {Remove the first element from $\hat{{ V}}_n$;}
		\IF {not $\hat{V}_n$}
		\STATE {Remove $n$ from $\mathcal{N}$;}
		\ENDIF
		\STATE {$J({\bf y}) = \{i| (y_i+1, y_{-i}) \in P_{\tilde{f}}, i\in \mathcal{N} \}$;}
		\ENDWHILE 
	\end{algorithmic}
\end{algorithm}
With this algorithm, the resources are allocated unit by unit. At each pass through, the algorithm maintains a list $J(\mathbf{y})$ of feasible agents who can still receive an additional unit of resource. It then assign the next unit to the agent from this set with the largest surrogate marginal utility.

Based on \cite{edmonds1971matroids,rado1957note,edmonds2003submodular}, \cite{glebov1973one} establishes the following theorem:
\begin{theorem} \label{thm:3}
	Suppose $\mathcal{J} \subseteq \mathbb{N}^N$ is a down-monotone and finite family of integer nonnegative vectors. The greedy algorithm optimizes a separable concave function over $\mathcal J$ if and only if $\mathcal J$ is an integral polymatroid.
\end{theorem}

Using this result, we can show the following theorem:
\begin{theorem}\label{thm:4}
	Given the surrogate marginal utility vectors $\hat{V}$, for any set of concave and non-decreasing utility functions $\bf {u}$ such that for each agent $n$, $u_n(\frac{y_n}{M}) = \sum_{m=1}^{y_n} \hat{v}_{nm} + \frac{\beta y_n}{M}$, then ${\bf y}^*({\hat{ V}})$ given by Algorithm \ref{greedy} is the optimizer to problem (\ref{5}) and the optimal outcome is $\text{OPT}( \hat{V})+ \beta$, i.e., 
	\begin{equation}\label{eq.17}
		{\bf y}^*({\hat{V}}) =\arg\max_{{\bf y} \in P_{\tilde{f}}}\sum_{n=1}^N {u}_{n}(\frac{y_n}{M}),
	\end{equation}
	\begin{equation}\label{eq:14}
		\max_{{\bf y} \in P_{\tilde{f}}}\sum_{n=1}^N {u}_{n}(\frac{y_n}{M}) =\sum_{n=1}^N {u}_{n}(\frac{y^*_n({\hat{V}})}{M}) =\text{OPT}( \hat{V})+ \beta  .
	\end{equation}
\end{theorem}	
\begin{proof}
	Combining Theorem \ref{thm:3} and Corollary \ref{lemma:2} yields (\ref{eq.17}). We can show (\ref{eq:14}) according to the definition of $\text{OPT}( \hat{V})$. The detailed proof is omitted here.
\end{proof}

This theorem shows the allocation determined by the greedy algorithm is the optimal solution to problem (\ref{5}) assuming the submitted surrogate marginal utilities are truthful. Next we show telling the truth is the (weakly) dominant strategy for each agent.

\begin{corollary}
	The quantized VCG mechanism is DSIC.
\end{corollary}
\begin{proof}
	The difference between marginal utility and surrogate utility is more of an implementation issue. There is a one to one mapping between them and having agents send one or the other is equivalent. The last term $\frac{y_n^*\beta}{M}$ in the payment compensates the difference ensuring the quantized VCG mechanism is still DSIC. 
\end{proof}

\subsection{Examples of Quantized VCG Performance}
In this section we examine the performance of the quantized VCG mechanism for the two examples introduced in the previous section. In the following section, we turn to analyzing the performance in more general settings. 

\textbf{Example 1 }(continued): Recall, for the original polymatroid $P_f$, the optimal allocation is ${\bf x}^* = (0.6,0.4)$. For $P_f$, we can find the minimum distance is $\bigtriangleup f = 0.2$. Hence, following Algorithm \ref{greedy}, we have $M \geq \frac{2}{0.2}=10$. Fig.~\ref{ex:1} shows the allocation $(x_1,x_2)$ given by the quantized VCG mechanism as $M$ varies.  Also shown are the optimal allocations in $P_f$. It can be seen that with 10 partitions, the quantized VCG allocation is close to the optimal and approaches it further as $M$ increases.
\begin{figure}[h]
	\centering
	\includegraphics[width=0.75\linewidth]{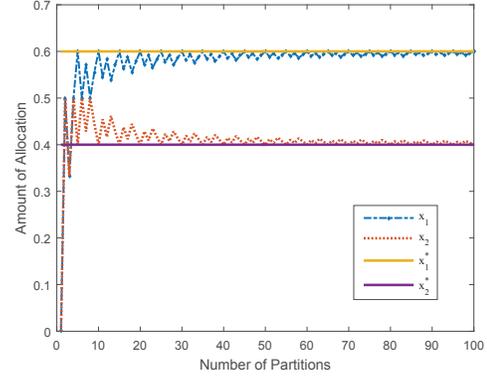}
	\caption{Performance of the quantized VCG mechanism for Example 1.}\label{ex:1}
\end{figure}%

\textbf{Example 2 }(continued): For this example, the optimal allocation is ${\bf x}^* = (0.7,0.2,0.1)$ and the minimum distance is $\bigtriangleup f = 0.5$. Hence, from Algorithm \ref{greedy} we have $M \geq \frac{2}{0.5}=4$. Fig. \ref{ex:1} shows the allocation $(x_1,x_2,x_3)$ as a function of the number of divisions. Again, the allocation given by the quantized VCG mechanism is close to the optimal allocation for $M = 4$ and becomes closer as $M$ increases.
\begin{figure}[h]
	\centering
	\includegraphics[width=0.75\linewidth]{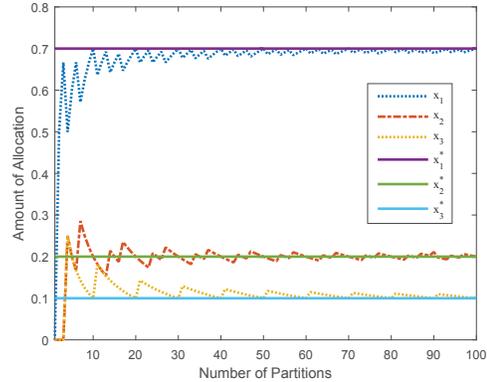}
	\caption{Performance of the quantized VCG mechanism for Example 2.}\label{ex:2}
\end{figure}%

Also these figures show the trade-off between the efficiency of the quantized VCG mechanism and the communication cost. As we increase the number of partitions, the difference between the allocation given by the quantized VCG mechanism and the optimal solution is smaller, giving better efficiency. On the other hand, a larger number of partitions means a larger communication cost (and larger computational costs).  

\subsection{Efficiency Analysis}
As in much of the literature, we evaluate the performance of the quantized VCG mechanism via the worst-case efficiency. This is defined as the ratio between the social welfare given by the quantized VCG mechanism and the optimal social welfare without quantization in the worst case (over the class of admissible utilities). Specifically, the worst-case efficiency for our mechanism is:
\begin{equation}
	\inf_{u_1,\cdots,u_N} \frac{\max_{{\bf y} \in P_{\tilde{f}}} \sum _{n=1}^N{u}_n(\frac{y_n}{M})}{\max_{{\bf x} \in P_f} \sum _{n=1}^N{u}_n(x_n)},
\end{equation}
where $u_i$ is any utility function satisfying Assumptions 2 and 3.  
\begin{theorem}\label{th:6}
	For a system with $N$ agents, suppose we partition the resource into $M$ divisions equally and run the quantized VCG mechanism as proposed, where $M$ satisfies the condition in Algorithm 2. In this case, the efficiency is at least $$\frac{ M\beta + 2(\alpha-\beta)}{M\alpha +2(\alpha - \beta) -[M-N+1]^+(\alpha-\beta)}.$$
\end{theorem}
\begin{proof}
	If $\bf x ^*$ is the optimal solution without quantization, then ${\bf y} = \lfloor {\bf x}^* M\rfloor$ must lie in $P_{\tilde{f}}$. In addition, we can always find $\bf \mathring y$ on the dominant face of $P_{\tilde{f}} - {\bf y}$, which is also a polymatroid by Lemma \ref{lemma:1} (and is also integral since {$\mathbf y$} is integer-valued). 
	Hence, 
	$\sum_{n=1}^N \mathring{y}_n = M -\sum_{n=1}^N \lfloor x^*_n M\rfloor=  \sum_{n=1}^N \{ x^*_n M\}$
	and so $\bf y + \mathring y$ is a feasible allocation of divisions on the dominant face of $P_{\tilde{f}}$.
	
	Thus, the efficiency is lower bounded by 
	\begin{equation}\label{eq:26}
		\frac{ \sum _{n=1}^N{u}_n(\frac{y_n+\mathring{y}_n}{M})}{ \sum _{n=1}^N{u}_n(x_n^*)}\geq  \frac{ \sum _{n=1}^N{u}_n(\frac{y_n}{M})+ \frac{\beta}{M}\sum_{n=1}^N \mathring{y}_n}{ \sum _{n=1}^N{u}_n(x_n^*)} .
	\end{equation}
	
	Recall that the utility functions are concave and for agent $n$ we have:
	\begin{align}
		\nonumber u_n(x_n^*) \leq& u_n'(\frac{y_n}{M})(x_n^*-\frac{y_n}{M})+ u_n(\frac{y_n}{M}) \\
		=&(u_n(\frac{y_n}{M}) - \frac{y_n}{M} u_n'(\frac{y_n}{M})) + x_n^* u_n'(\frac{y_n}{M}).
	\end{align}
	Therefore, 
	\begin{align}
		\nonumber &\frac{ \sum _{n=1}^N{u}_n(\frac{y_n}{M})+ \frac{\beta}{M}\sum_{n=1}^N \mathring{y}_n}{ \sum _{n=1}^N{u}_n(x_n^*)} . \\
		\nonumber \label{eq:1818}\geq &\frac{\sum_{n=1}^N \frac{y_n}{M} u_n'(\frac{y_n}{M})+\frac{\beta}{M} \sum_{n=1}^N \mathring{y}_i }{\sum_{n=1}^N x_n^* u_n'(\frac{y_n}{M})}\\
		= &\frac{\sum_{n=1}^N \lfloor x^*_n M\rfloor u_n'(\frac{y_n}{M}) +\beta \sum_{n=1}^N \{ x^*_n M\}}{\sum_{n=1}^N \lfloor x^*_n M\rfloor u_n'(\frac{y_n}{M})+\sum_{n=1}^N \{ x^*_n M\} u_n'(\frac{y_n}{M})}.
	\end{align}
	This indicates to find the lower bound for the efficiency, we should only focus on utility functions $u_n$ that are linear for $x \leq \frac{y_n}{M}$.
	
	Proposition~\ref{prop:2} shows that $\tilde{f}(S) \geq 2$ for any subset $S$. Hence, the greedy algorithm must pick the 2 greatest bids. Let $\alpha^* = \max_n  u_n'(\frac{y_n}{M})$ and so we have:
	\begin{equation}
		\sum_{n=1}^N \lfloor x^*_n M\rfloor u_n'(\frac{y_n}{M}) \geq 2\alpha^* + \beta (\sum_{n=1}^N \lfloor x^*_n M\rfloor-2),
	\end{equation}
	\begin{equation}
		\sum_{n=1}^N \{ x^*_n M\} u_n'(\frac{y_n}{M}) \leq \alpha^* \sum_{n=1}^N \{ x^*_n M\}.
	\end{equation}
	Therefore,
	\begin{equation}\label{eq:28}
		(\ref{eq:1818})\geq \frac{(M-2)\beta + 2\alpha^*}{(M+2)\alpha^* -2\beta -\sum\limits_{n=1}^N \lfloor x^*_n M\rfloor(\alpha^*-\beta)},
	\end{equation}
	which is increasing in $\sum_{n=1}^N \lfloor x^*_n M\rfloor$ because $\alpha^* -\beta $ is always nonnegative. Since $\sum_{n =1}^N x_n^*M = M $, we have
	$
	\sum\limits_{n=1}^N \{ x^*_n M\} \leq N-1,
	$
	and hence
	\begin{equation}\label{eq:29}
		\sum_{n=1}^N \lfloor x^*_n M\rfloor \geq  \max\{ 0,M - N + 1\} .
	\end{equation}
	Substituting (\ref{eq:29}) into (\ref{eq:28}) yields:
	$$  (\ref{eq:1818}) \geq\left\{
	\begin{aligned}
	&\frac{2\alpha^* +(M-2)\beta}{(M+2)\alpha^* -2\beta},  \quad \quad  \quad \quad \quad \quad M < N-1, \\ 
	&\frac{2\alpha^* + (M-2)\beta}{(N+1)\alpha^* + (M-N-1)\beta},   \ \ \  M \geq N-1.
	\end{aligned}
	\right.$$
	In either case, we can check that the lower bound is a decreasing function over $\alpha^*$ and we know $\alpha^* \leq \alpha$, therefore, we finish the proof by letting $\alpha^* = \alpha$.
\end{proof}

This lower bound is tight. Consider the example with two agents: $P_f = \{{\bf x}| x_1\leq 1-\eta, x_2 \leq \frac{2}{3} +\eta, x_1+x_2 \leq 1\}$ and $u_1(x) = \alpha x$, $u_2(x) =\beta x$. We can see in this case $\bigtriangleup f = \frac{2}{3}$ and $M$ should be greater than 3. If we pick $M =3 $, then the resulting allocation is $(\frac{2}{3}, \frac{1}{3})$, while the optimal allocation is $(1-\eta,\eta)$.  Hence, the efficiency is $\frac{2\alpha +\beta}{3(1-\eta)\alpha + 3\eta\beta}$. As $\eta$ goes to 0, we achieve the lower bound, which is $\frac{2\alpha+\beta}{3\alpha}$.

\textit{Remark 2}: When $\beta \rightarrow \alpha $, the lower bound approaches 1 and is tight, which makes sense because in this case, all agents have the same valuations of the resource and so the allocation is always efficient.

The lower bound is an increasing function of number of partitions, which indicates the trade-off between the efficiency and the communication cost as well. As $M$ goes to infinity, the lower bound goes to 1 since the mechanism approaches the VCG mechanism. Moreover, for a fixed number of partitions, the lower bound decreases as the number of agents grows, and when $N \geq M + 1$, the lower bound becomes independent of the number of agents and is a constant given $M$.

\section{Rounded Quantized VCG Mechanism}
In the last section, we specified the quantized VCG mechanism for allocating a divisible resource in a polymatroid environment and provided a lower bound for the worst-case efficiency. In the quantized VCG mechanism, agent $n$ only needs to send $\tilde{f}(\{n\})$ bids in total to indicate her utility for the possible allocation. This is a reduction of the dimensionality of the needed communication - a common measure used in both the engineering literature (e.g. \cite{yang2007vcg,haoge2018}) and in economics (e.g., this is the notion of information efficiency used by Hurowitz and Reiter \cite{hurwicz2006designing}). However, from an information theoretic point of view, conveying a real number still requires an infinite number of bits. In a small networked system, agents can send a very long bids to approximate the real value, but when the number of agents grows in the network, the total amount of communication may be unacceptable. Hence, in this section we further consider quantizing the bids sent by each agent in addition to the quantization  of the resource in a large networked system (e.g., $N \geq 10$).\footnote{Again this builds on work in \cite{haoge2018}, which considered a similar approach for a single divisible resource and only one of the bidding strategies we discuss below.} 

Concretely, we determine a monetary unit $\delta$. Each agent is restricted to give valuations that are integer multiple of $\delta$. Thereby, to indicate the surrogate marginal utility, agent $n$ needs to send an integer vector $\hat{W}_n = \{\hat{w}_{n1},\cdots,\hat{w}_{n\tilde{f}(\{n\})} \}$ to approximate her true utility. We call the resulting mechanism the \textit{rounded quantized VCG mechanism. }

More precisely this mechanism is defined as follows:
\begin{itemize}
	\item{Determine the number of partitions and partition the resource into $M$ divisions equally using Algorithm \ref{integer-valued}. Again, denote the feasible region for number of units by $P_{\tilde{f}}$}.
	\item{Determine and broadcast the monetary unit $\delta$ }.
	\item{Solicit and accept sealed value vectors $\hat{W}_n $}.
	\item{Run the quantized VCG mechanism with marginal utility vectors $\hat{V}_n= \hat{W}_n \delta$. Break ties randomly and determine the allocation and payment.}
\end{itemize}

In other words, following this mechanism, agent $n$ equivalently reports another function $\tilde{u}_n(x)$ to approximate $u_n$, where 
\begin{equation}
	\tilde{u}_n(\frac{y_n}{M}) = \sum _{m=1} ^{y_n} \hat{w}_{nm}\delta + \beta \frac{y_n}{M}.
\end{equation}
As shown in Theorem \ref{thm:4}, the resource allocator aims to find allocation $\bf \tilde{y}^*$ to maximize the sum of $\tilde{u}_n$, e.g.,
\begin{equation}\label{eq:33}
	{\bf \tilde{y}}^* = \arg\max_{{\bf y } \in P_{\tilde{f}}} \sum_{n=1}^N \tilde{u}_n(\frac{y_n}{M}).
\end{equation}

Obviously, the social welfare given by the rounded quantized VCG mechanism is less than that given by the quantized VCG since each agent cannot report her utility accurately. We will analyze the loss due to this restriction later in this section. Also,  here we require that agents submit a non-increasing sequence of bids (as we will see in the following, due to bid quantization agents may not have an incentive to do this without the restriction). The main reason for this restriction is that it is needed for us to use the greedy algorithm to correctly determine the outcome and payments. 
\subsection{Equilibrium Analysis}
It is shown in \cite{haoge2018} that there is no dominant strategy when using such a mechanism over a single-link network. Hence, this is also true for our more general polymatroid setting. \footnote{If we set all constraints to the same value, e.g., $f(S) = 1$, for all $S$, then our problem is reduced to a single-link problem.} We instead adopt the following relaxed solution concept which allows for agents to tolerate some loss:

\begin{definition}
	Given any $\epsilon \geq 0$, a strategy $\sigma_n^*$ for agent $n$ is called \textit{$\epsilon$-dominant} if for all $\sigma_n$ and $\sigma_{-n}$, $u_i(\sigma_n,\sigma_{-n})  -u_n(\sigma_n^*,\sigma_{-n}) \leq\epsilon.$
\end{definition}

In other words, for agent $n$, any unilateral deviation from strategy $\sigma_n^*$ leads to at most an $\epsilon$ gain. A dominant strategy is a special case of an $\epsilon$-dominant strategy with $\epsilon = 0$.

\begin{definition}
	A strategy profile $\sigma^*= (\sigma_1^*,\cdots,\sigma_N^*)$ forms an $\epsilon$-equilibrium if for every agent $n$, $\sigma_n^*$ is $\epsilon$-dominant. 
\end{definition}

\begin{figure*}[t]
	\centering
	\begin{subfigure}[b]{0.3\textwidth}
		\centering
		\includegraphics[width=0.9\linewidth]{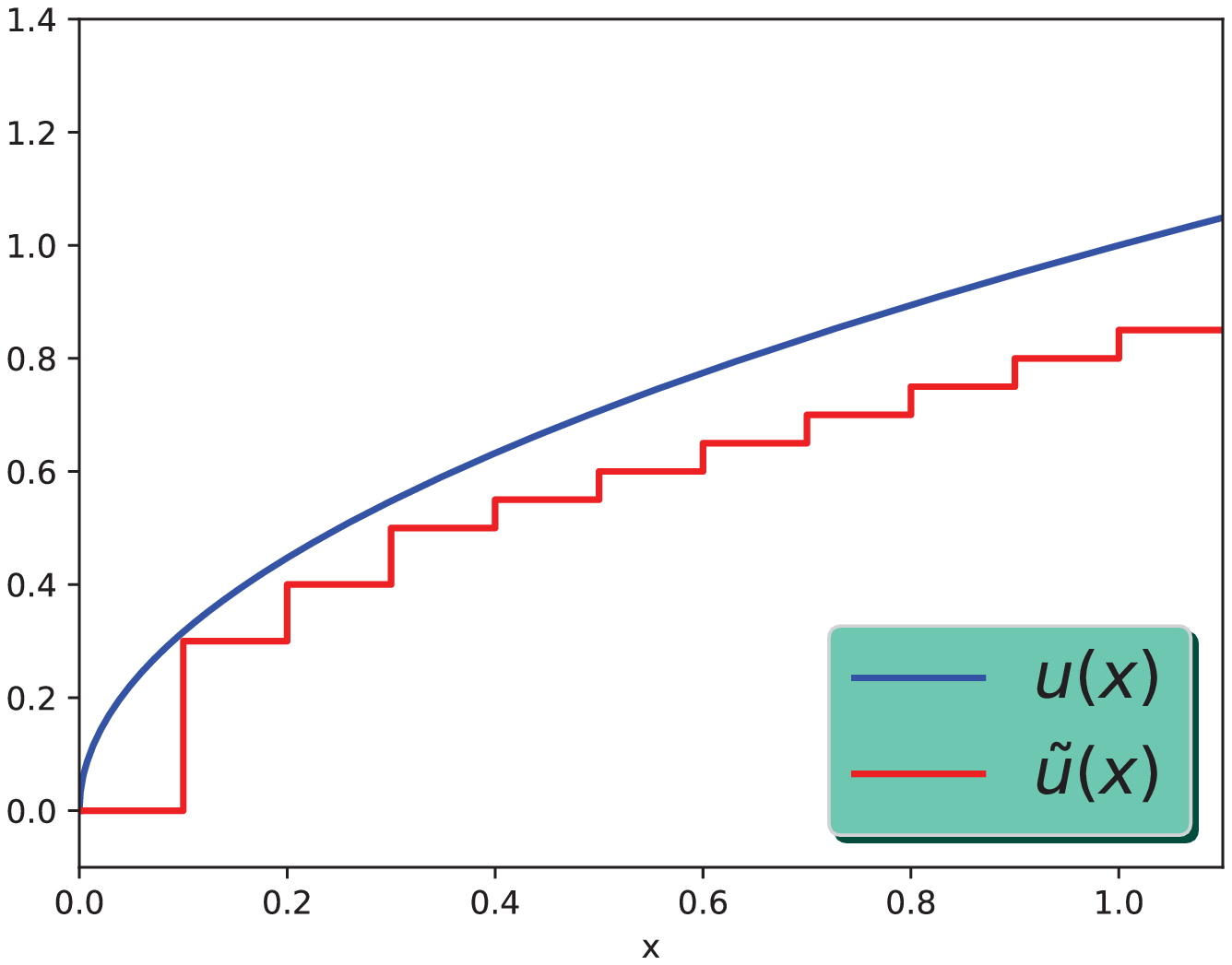}\label{fig:floor}
		\caption{The FLOOR strategy.}
	\end{subfigure}%
	\begin{subfigure}[b]{0.3\textwidth}
		\centering
		\includegraphics[width=0.9\linewidth]{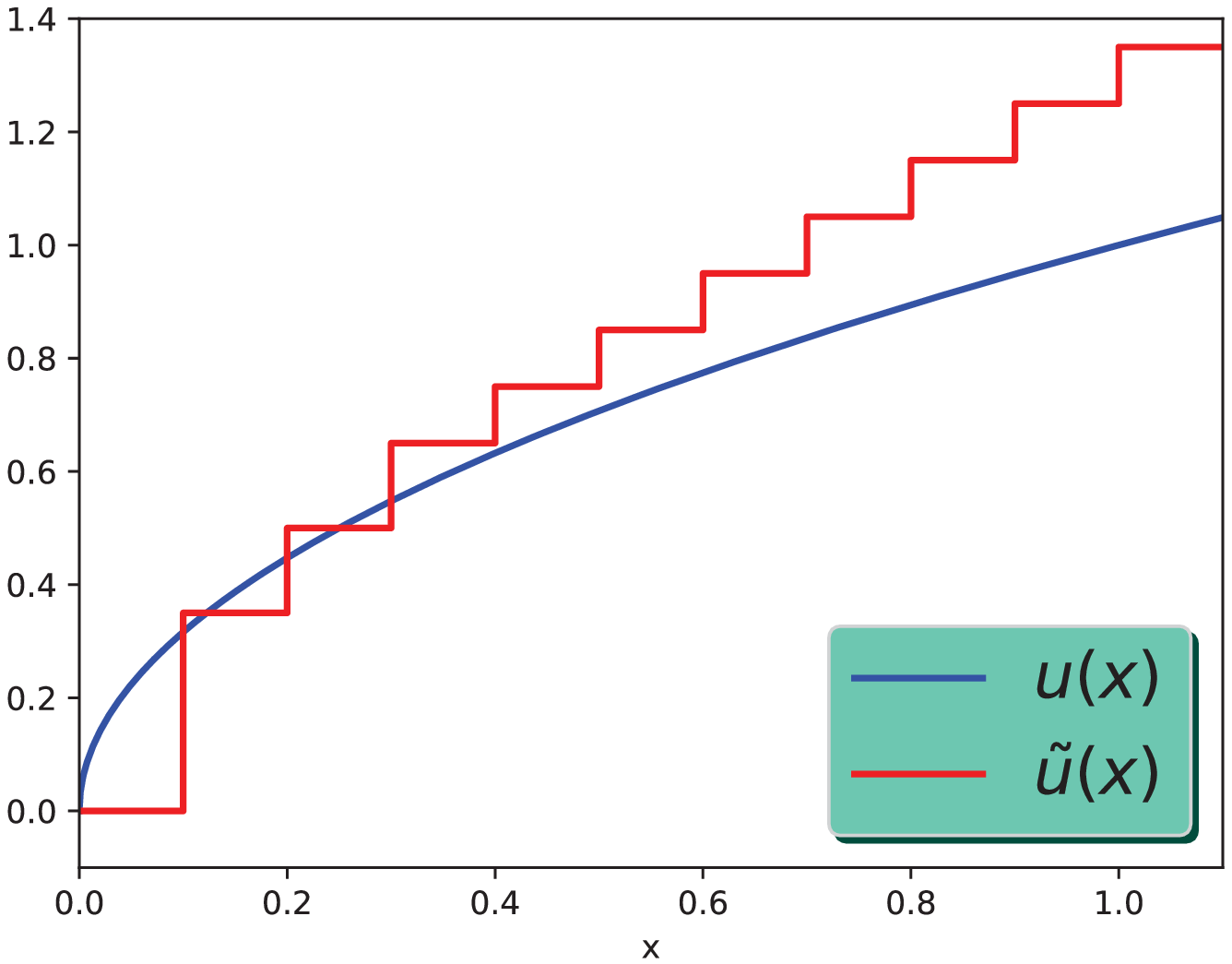}
		\caption{The CEILING strategy.}
	\end{subfigure}%
	\begin{subfigure}[b]{0.3\textwidth}
		\centering
		\includegraphics[width=0.9\linewidth]{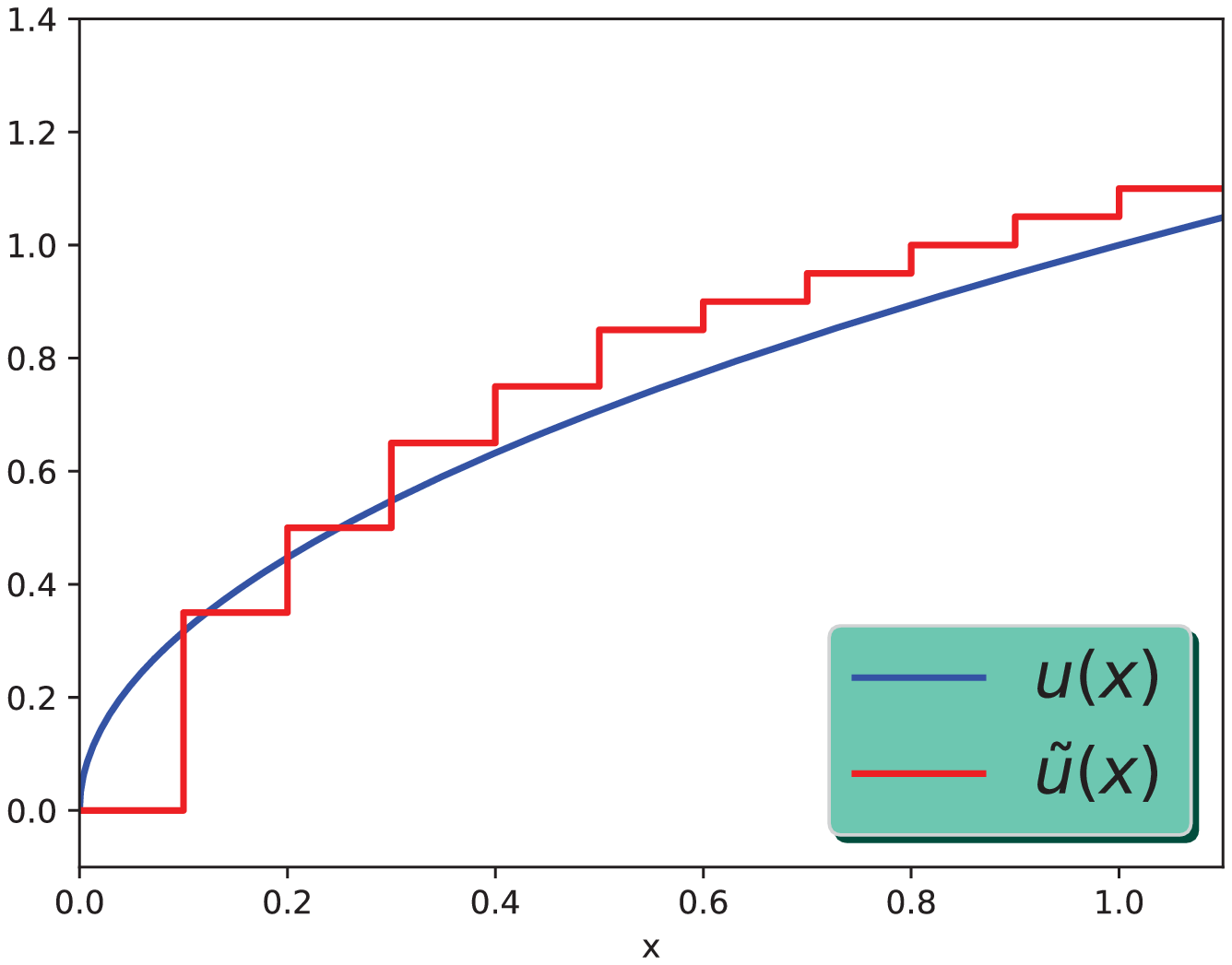}
		\caption{The CEILOOR strategy.}
	\end{subfigure}
	\caption{Equivalent utility function $\tilde{u}(x)$ under different strategies, the three example equilibrium strategies, where the true utility function $\hat{u}(x) = \sqrt{x}$ and the monetary unit $\delta = 0.05$.}\label{fig:animals} 
\end{figure*}

For this section, we assume agents can tolerate a loss $\epsilon$ and adopt the above solution concept. Note that this solution concept shares the key properties of dominant strategies. For example, once again, agents do not need any knowledge of the actions or rationality of other players to determine their action. 

One result of this relaxation in the solution concept is that agents no longer have a unique strategy. We illustrate this in the next three theorems, which each show a different equilibrium profile (see also Figure \ref{fig:animals}).
\begin{theorem}\label{th:666}
	Under the rounded quantized VCG mechanism with given $\epsilon$, if set $\delta = \frac{\epsilon}{M}$, then for any agent $n$, it is an $\epsilon$-dominant strategy to report $\hat{w}_{nm}=\lfloor \frac{\hat{v}_{nm}}{\delta}\rfloor$ for the $m^{th}$ partition. Moreover, if all agents play this $\epsilon$-dominant strategy, the maximum difference between social welfare given by the quantized VCG and the rounded quantized VCG is $\eta = \epsilon=M\delta $.
\end{theorem}
\begin{proof}
	For agent $n$, as shown in the last section, reporting the true surrogate marginal utility $\hat{V}_n$ (equivalently reporting ${u}_n$) is optimal without the quantized bid constraint and the corresponding maximum payoff is: 
	\begin{align}
		\nonumber&\max\limits_{{\bf y} \in P_{\tilde{f}}} u_n(\frac{y_n}{M}) +\sum_{j \neq n}\tilde{u}_j(\frac{y_j}{M}) -\text{OPT}(0,\hat{W}_{-n}\delta)-\beta \\
		\label{17}=& R -\text{OPT}(0,\hat{W}_{-n}\delta)-\beta.
	\end{align}
	Here, $\hat{W}_{-n}$ indicates the other agents' bids, which correspond to $\tilde{u}_{-n}$, and 
	$
	R =\max\limits_{{\bf y} \in P_{\tilde{f}}} u_n(\frac{y_n}{M}) +\sum_{j \neq n}\tilde{u}_j(\frac{y_j}{M}).
	$
	Suppose agent $n$ reports $W_n = \lfloor \frac{V_n}{\delta} \rfloor$ and the corresponding utility function is $\tilde{u}_n$. Given the final allocation $\bf \tilde{y}^*$, then agent $n$'s payoff is:
	\begin{align}
		\nonumber u_n(\frac{\tilde{y}_n^*}{M}) - p_n =& u_n(\frac{\tilde{y}_n^*}{M}) +\sum_{j \neq n}\tilde{u}_j(\frac{\tilde{y}_j^*}{M}) -\text{OPT}(0,\tilde{W}_{-n}\delta)-\beta\\
		\label{16}= & L - \text{OPT}(0,\tilde{W}_{-n}\delta) -\beta,
	\end{align}
	where 
	$
	\label{eq:L}L= \max \limits_{{\bf y} \in P_{\tilde{f}}} \sum\limits_{n=1}^N \tilde{u}_{n}(\frac{y_n}{M})+u_n(\frac{\tilde{y}_n^*}{M}) -\tilde{u}_n(\frac{\tilde{y}_n^*}{M}).
	$
	
	To show the strategy is $\epsilon-$dominant, it is sufficient to show $L \geq R -\epsilon$, which implies that agent $n$'s loss under this strategy is less than $\epsilon$ compared with the maximum payoff she can get.
	
	As shown in Fig. \ref{fig:animals}(a), by reporting the floor value, for $y_n =0,1\cdots, M$,
	\begin{equation}\label{eq:36}
		\tilde{u}_n(\frac{y_n}{M}) = \sum _{m=1} ^{y_n} \hat{w}_{nm}\delta +\beta \frac{y_n}{M}= \sum _{m=1} ^{y_n} \lfloor \frac{\hat{v}_{nm}}{\delta}\rfloor\delta +\beta \frac{y_n}{M}.
	\end{equation}
	Hence, we have:
	\begin{equation}\label{eq:37}
		u_n(\frac{y_n}{M}) -y_n\delta \leq \tilde{u}_n(\frac{y_n}{M})   \leq u_n(\frac{y_n}{M})
	\end{equation} 
	and so 
	\begin{align}
		\nonumber L \geq &\max \limits_{{\bf y} \in P_{\tilde{f}}} \sum_{n=1}^N \tilde{u}_{n}(\frac{y_n}{M})\\
		\nonumber\geq & \max\limits_{{\bf y} \in P_{\tilde{f}}}u_n(\frac{y_n}{M}) +\sum_{j \neq n}\tilde{u}_j(\frac{y_j}{M}) - y_n\delta\\
		\label{20}\geq & R- M\delta.
	\end{align}
	Therefore, it is an $\epsilon$-dominant strategy to report the floor value.
	
	Next, we show that if all agents follow this strategy, the difference between the social welfare given by the quantized VCG and the rounded quantized VCG mechanism is upper bounded by $M\delta $. In other words, the loss in social welfare due to the quantized bid constraint is no greater than $M\delta$. According to the definition, the difference is:
	\begin{align*}
		\sum_{n=1}^N u_{n}(\frac{y_n^*}{M}) - \sum_{n=1}^N u_{n}(\frac{\tilde{y}^*_n}{M}) \leq & \sum_{n=1}^N \tilde{u}_{n}(\frac{y_n^*}{M})+ y_n^*\delta- \sum_{n=1}^N u_{n}(\frac{\tilde{y}^*_n}{M})\\
		\leq & \sum_{n=1}^N \tilde{u}_{n}(\frac{\tilde{y}^*_n}{M})- \sum_{n=1}^N u_{n}(\frac{\tilde{y}^*_n}{M}) + M\delta \\
		\leq & M\delta .
	\end{align*}
	The inequality (\ref{eq:37}) is used in the last two steps. 
\end{proof}

For simplicity, we call this the FLOOR strategy. Similarly, we next define a CEILING strategy and again bound for loss in social welfare.
\begin{theorem}
	Under the rounded quantized VCG mechanism with given $\epsilon$, if set $\delta = \frac{\epsilon}{M}$, then for any agent $n$, it is an $\epsilon$-dominant strategy to report $\hat{w}_{nm}=\lceil \frac{\hat{v}_{nm}}{\delta}\rceil$ for the $m^{th}$ partition. If all agents play this $\epsilon$-dominant strategy, the maximum difference between the social welfare given by the quantized VCG and the rounded quantized VCG is $\eta =\epsilon= M\delta $.
\end{theorem}

	The proof is similar as that for FLOOR strategy.

By reporting the floor or ceiling value, each agent has an approximation for the true utility function and makes their bids non-decreasing, as we require. However, the bid is always smaller (greater) than the true value if we use the FLOOR (CEILING) strategy. As we can observe in Fig \ref{fig:animals}(a)-(b), the difference between bids and true values will accumulate and the approximation is less accurate when $x$ is large.

Inspired by the above two strategies, to give a better approximation of the true utility function, one may think about combining the two strategies instead of reporting floor values or ceiling values solely. This is our third $\epsilon$-dominant strategy, which we refer to as the CEILOOR strategy. 
\begin{theorem}
	Under the rounded quantized VCG mechanism with given $\epsilon$, if the number of partitions $M$ is even and set $\delta = \frac{\epsilon}{M}$, then it is an $\frac{\epsilon}{2}$-dominant strategy for agent $n$ to report $$ \hat{w}_{nm}=\left\{
	\begin{aligned}
	\lceil \frac{\hat{v}_{nm}}{\delta}\rceil \ \ \ \ \ \ \  &0 < m \leq \frac{M}{2} \\
	\lfloor \frac{\hat{v}_{nm}}{\delta}\rfloor\ \ \ \ \ \ \  &\frac{M}{2} < m \leq M. 
	\end{aligned}
	\right.$$
	If all agents play this strategy, the maximum difference between the social welfare given by the quantized VCG and the rounded quantized VCG is $\eta =\epsilon= M\delta $.
\end{theorem}
The detailed proof is omitted here.

\textit{Remark 3:} Though all these strategies lead to no more than $\epsilon$ regret, note that only the FLOOR strategy is individually rational as in the other cases an agent's pay-off may be negative. However, if we also relax individual rationality to allowing for a loss of at more $\epsilon$, then all strategies are $\epsilon$-individually rational. 

In the previous results, we assumed that all agents choose the same type of $\epsilon$-dominant strategy. Next, we show that if agents choose different types of strategies, then this can lead to larger welfare losses. 

\begin{theorem}
	Under the rounded quantized VCG mechanism with given $\epsilon$, if $\delta = \frac{\epsilon}{M} $ and there is no agreement among agents on the type of $\epsilon$-dominant strategy, then the maximum difference between social welfare given by the quantized VCG and the rounded quantized VCG is $\eta = 2\epsilon $.
\end{theorem}

We omit the proof here for space consideration.

The preceding analysis showed that the inconsistence in the choice of strategy may result in a loss of efficiency. This suggest that a resource allocator may want to encourage agents to follow the same type of strategy. Also note that the CEILOOR strategy is more appealing because it offers a smaller bound on loss for each agent (since it is $\frac{\epsilon}{2}$-dominant instead of $\epsilon$ dominant) and thus might be preferable. On the other hand, as noted above the FLOOR strategy is individually rational and so might be preferred from that point-of-view.

Under the assumption that all agents follow the same strategy, we hope to analyze the lower bound for efficiency in each case. However, before divining in, one important fact should be considered that in practical setting, compared with the pay-off, the loss each agent can tolerate should be pretty small. Hence, we further make the following assumption to ensure that $\epsilon$ is small compared to the possible utility an agent may obtain.
\begin{assumption}
	For each agent, the maximum loss $\epsilon$ is smaller than half of the minimum utility if she gets all the resource, i.e., $\epsilon \leq \frac{\beta}{2}$.
\end{assumption} 

This assumption provides a pretty loose upper bound for $\epsilon$ and in fact $\epsilon$ could be much smaller than this bound in practice. Based on the previous results, we next derive the overall worst-case efficiency bounds in different scenarios. 
\begin{theorem}\label{thm:11}
	For a large networked system with more than 3 agents, under the rounded quantized VCG mechanism with given $\epsilon$, if set $\delta = \frac{\epsilon}{M}$ and all agents choose the same strategy simultaneously, the worst-case efficiency is no less than $$\frac{ M\beta + 2(\alpha-\beta)-M\epsilon}{M\alpha +2(\alpha - \beta) -[M-N+1]^+(\alpha-\beta)}.$$
\end{theorem}
\begin{proof}
	The main idea is simply combining the losses from quantizing the resource and the bids. The details are omitted here.
\end{proof}

\subsection{Discussion}
For a specific scenario, given a requirement for the efficiency and $\epsilon$, these results can be used to determine a minimum number of partitions that we can guarantee will achieve this requirement. Let $M^*$ denote the value for a given bound which meets this target. The number of partitions a planner should choose to minimize the communication cost is then be $\max(M^*, \frac{2}{\triangle f})$ where the other terms follow from Algorithm 1.

Likewise, given the number of partitions, we can study the impact of the monetary unit $\delta$. Under the three $\epsilon$-dominant strategies discussed, the maximum bid one may submit is $\lceil \frac{\alpha-\beta}{M\delta}\rceil = \lceil \frac{\alpha-\beta}{\epsilon}\rceil$, and one agent needs at most $M\log_2 \lceil \frac{\alpha-\beta}{\epsilon}\rceil$ bits in total to convey its bids to the resource allocator. We can see the trade-off between communication cost and allocation efficiency here. If agents have a lower tolerance of loss, which means $\epsilon$ decreases, then from Theorem \ref{thm:11}, we know the worst-case efficiency will improve but the communication cost will increase. As $\epsilon$ goes to 0, the rounded quantized VCG mechanism approaches the quantized VCG mechanism. On the other hand, a large $\epsilon$ indicates the agents will only roughly approximate their utility function; subsequently, the worst-case efficiency will be low but so will the required communication cost. 

\section{Conclusion}
We considered two mechanisms for allocating a resource constrained to lie in a polymatroid capacity region with limited communication exchanged: the quantized VCG mechanism and the rounded VCG mechanism. These two mechanisms utilize quantization to reduce the communication cost between the resource allocator and the agents. Using the properties of polymatroids, we showed that these mechanisms preserve the dominant strategy incentive properties of VCG to varying degrees. We also bounded the worst-case efficiency for each mechanism in different scenarios. There are many ways that this work could be extended including allowing collusions among agents or considering revenue maximization.

\bibliographystyle{unsrtnat}
\bibliography{sample-bibliography1}
\appendix

\end{document}